%% file: main.tex
\pdfoutput=1 
\documentclass{article}

 \usepackage[margin=1in]{geometry}
 \usepackage{parskip}
 \usepackage{microtype}
 \let\iclrfinalcopy\relax

\usepackage{amssymb}
\usepackage{amsmath}
\usepackage{graphicx}
\usepackage{mathtools}
\usepackage{needspace} 
\usepackage{multicol}
\usepackage{listings}
\usepackage{amsthm}
\usepackage[shortlabels]{enumitem}
\usepackage[italicdiff]{physics}
\usepackage{cancel}
\usepackage{verbatim}
\usepackage{comment}
\usepackage{array}
\usepackage{tikz-cd}
\usepackage{import}
\usepackage{booktabs}

\usepackage{comment}
\usepackage{array}
\usepackage[colorlinks, citecolor=blue]{hyperref}
\usepackage{cleveref}
\usepackage{forest}
\usepackage{float}
\usepackage{xspace}
\usepackage{aliascnt}
\usepackage{algorithm}
\usepackage[noend]{algpseudocode}
\usepackage{wrapfig}
\usepackage{thm-restate}
\usepackage{adjustbox}

\usepackage{autonum} %
\makeatletter
\autonum@generatePatchedReferenceCSL{Cref}
\autonum@generatePatchedReferenceCSL{cref}
\makeatother

\usepackage{natbib}
\makeatletter
\def\NAT@spacechar{~}
\makeatother

\newcommand{\delimit}[3]{\newcommand{#1}[1]{\left#2##1\right#3}}
\delimit \ceil \lceil \rceil
\delimit \floor \lfloor \rfloor

\DeclareMathOperator*{\E}{\mathbb E}
\DeclareMathOperator*{\co}{co}

\let\op\operatorname
\let\eps\varepsilon

\let\oldboxed\boxed
\renewcommand{\boxed}[1]{\oldboxed{#1}\,}

\newcommand{\R}{\mathbb R}

\newcommand{\poly}{\op{poly}}
\newcommand{\ie}{{\em i.e.}\xspace}
\newcommand{\eg}{{\em e.g.}\xspace}
\newcommand{\Reg}{\op{Reg}}
\newcommand{\dPhi}[1]{\Phi_\textsc{#1}}
\let\vec\boldsymbol
\let\mat\mathbf

\newcommand{\Root}{{\varnothing}}
\let\tilde\widetilde
\let\hat\widehat
\let\ip\ev

\usepackage{centernot}

\renewcommand\grad\nabla

\let\mc\mathcal

\let\cite\citep

\newtheorem{theorem}{Theorem}
\newtheorem*{theorem*}{Theorem}
\numberwithin{theorem}{section}

\newtheorem{lemma}[theorem]{Lemma}

\newtheorem{corollary}[theorem]{Corollary}

\theoremstyle{definition}
\newtheorem{definition}[theorem]{Definition}

\newcommand{\commentsymbol}{\it\color{gray}$\triangleright$~}
\algrenewcommand\algorithmiccomment[1]{\hfill{\commentsymbol#1}}

\algdef{SE}[SUBALG]{Indent}{EndIndent}{}{\algorithmicend\ }%
\algtext*{Indent}
\algtext*{EndIndent}

\newcommand{\changed}[1]{#1}

\newcommand{\pone}{{\ensuremath{\color{p1color}\blacktriangle}}\xspace}
\newcommand{\ptwo}{{\ensuremath{\color{p2color}\blacktriangledown}}\xspace}
\newcommand{\util}[2]{\textbf{\textsf{{\color{p#1color}#2}}}}
\newcommand{\action}[3]{\textbf{\textsf{\color{p#1color}#2$_{\boldsymbol{#3}}$}}}

\definecolor{p0color}{RGB}{0,0,0}
\definecolor{p1color}{RGB}{31,119,180}
\definecolor{p2color}{RGB}{255,127,14}

\definecolor{s2color}{RGB}{44,160,44}
\definecolor{p3color}{RGB}{214,39,40}

\title{Mediator Interpretation and Faster Learning Algorithms for Linear Correlated Equilibria in General Extensive-Form Games}

\iclrfinalcopy
\author{%
Brian Hu Zhang \\
Carnegie Mellon University\\
\texttt{bhzhang@cs.cmu.edu} \\
\and
Gabriele Farina \\
MIT\\
\texttt{gfarina@mit.edu} \\
\and
Tuomas Sandholm \\
Carnegie Mellon University\\
Strategic Machine, Inc. \\
Strategy Robot, Inc. \\
Optimized Markets, Inc. \\
\texttt{sandholm@cs.cmu.edu} \\
}\begin{document}

\maketitle
\begin{abstract} %
    A recent paper by \citet{Farina23:Polynomial} established the existence of uncoupled no-linear-swap regret dynamics with polynomial-time iterations in extensive-form games. The equilibrium points reached by these dynamics, known as \emph{linear correlated equilibria}, are currently the tightest known relaxation of correlated equilibrium that can be learned in polynomial time in any finite extensive-form game. However, their properties remain vastly unexplored, and their computation is onerous. In this paper, we provide several contributions shedding light on the fundamental nature of linear-swap regret. First, we show a connection between linear deviations and a generalization of {\em communication deviations} in which the player can make queries to a ``mediator'' who replies with action recommendations, and, critically, the player is not constrained to match the timing of the game as would be the case for communication deviations. We coin this latter set the {\em untimed communication (UTC) deviations}. We show that the UTC deviations coincide precisely with the linear deviations, and therefore that any player minimizing UTC regret also minimizes linear-swap regret. We then leverage this connection to develop state-of-the-art no-regret algorithms for computing linear correlated equilibria, both in theory and in practice. In theory, our algorithms achieve polynomially better per-iteration runtimes; in practice, our algorithms represent the state of the art by several orders of magnitude.
\end{abstract}

\section{Introduction}

In no-regret learning, a player repeatedly interacts with a possibly adversarial environment. The task of the player is to minimize its {\em regret}, which is defined to be the difference between the utility experienced by the player, and the largest utility that it could have achieved in hindsight if it had played other strategies instead, according to some {\em strategy transformation} or {\em deviation}. The set of allowable deviations defines the notion of regret, with larger sets corresponding to tighter notions. Two extremes are {\em external deviations}, which are the set of all constant transformations, and {\em swap deviations}, which are  all possible functions.

In games, no-regret learning has a very tight connection to notions of {\em correlated equilibrium}. Each notion of regret has its corresponding notion of equilibrium, which will be reached by a set of players that independently run no-regret algorithms for that notion of regret. External and swap deviations, respectively, correspond to the well-known {\em normal-form coarse correlated} equilibrium~\cite{Moulin78:Strategically} (NFCCE) and {\em normal-form correlated} equilibrium~\cite{Aumann74:Subjectivity} (NFCE). For extensive-form games specifically, other sets of deviations include the {\em trigger deviations}~\cite{Gordon08:No,Farina22:Simple}, which correspond to {\em extensive-form correlated equilibrium}~\cite{Stengel08:Extensive}, and the {\em communication deviations}~\cite{Zhang22:Polynomial,Fujii23:Bayes}, which correspond to {\em communication equilibrium}\footnote{Technically, communication equilibria are more broad than $(\Phi_i)$-equilibria where $\Phi_i$ is the set of communication deviations: in a communication equilibrium, there is an explicit {\em mediator} who has the power not only to sample a strategy profile, but also to {\em pass private information} between players---so a communication equilibrium is not necessarily a correlated profile at all. This distinction is fairly fundamental: it is the reason why polynomial-time algorithms for optimal communication equilibrium can exist for extensive-form games~\mbox{\cite{Zhang22:Polynomial}}. We may call a communication equilibrium that also happens to be a correlated profile a {\em private communication equilibrium}, where {\em private} denotes that the mediator is not allowed to pass information between players.  However, since this paper focuses on no-regret learning, we have no reason to make this distinction, so we largely ignore it.\label{footnote:comm}}~\cite{Myerson86:Multistage,Forges86:Approach}.

In this paper, we consider a notion of regret first studied for extensive-form games by \citet{Farina23:Polynomial}, namely, regret with respect to the set of {\em linear functions} from the strategy set to itself. This notion is a natural stepping stone between external regret, which is very well studied, and swap regret, for which \changed{achieving $\poly(d) \cdot T^c$ regret, where $d$ is the size of the decision problem and $c<1$, is a long-standing open problem}. We make two main contributions.

The first contribution is conceptual: we give, for extensive-form games, an {\em interpretation} of the set of linear deviations. More specifically, we will first introduce a set of deviations, which we will call the {\em untimed communication (UTC) deviations} that, a priori, seems very different from the set of linear deviations at least on a conceptual level. The deviation set, rather than being defined {\em algebraically} (linear functions), will be defined in terms of an interaction between a {\em deviator}, who wishes to evaluate the deviation function at a particular input, and a {\em mediator}, who answers queries about the input. We will show the following result, which is our first main theorem:
\begin{theorem*}
    The untimed communication deviations are precisely the linear deviations.
\end{theorem*}
The mediator-based framework is more in line with other extensive-form deviation sets---indeed, all prior notions of regret for extensive form, to our knowledge, including all the notions discussed above, can be expressed in terms of the framework. As such, the above theorem places linear deviations firmly within the same framework usually used to study deviations in extensive form.

We will then demonstrate that the set of UTC deviations is expressible in terms of {\em scaled extensions}~\cite{Farina19:Efficient}, opening up access to a wide range of extremely fast algorithms for regret minimization, both theoretically and practically, for UTC deviations and thus also for linear deviations. Our second main theorem is as follows.
\begin{theorem*}[Faster linear-swap regret minimization]
    There exists a regret minimizer with regret $O(d^2\sqrt{T})$ against all linear deviations, and whose per-iteration complexity is dominated by the complexity of computing a fixed point of a linear map $\phi^{(t)} : \co\mc X \to \co\mc X$.
\end{theorem*}

In particular, using the algorithm of \citet{Cohen21:Solving} to solve the linear program of finding a fixed point, our per-iteration complexity is $\tilde O(d^\omega)$, where $\omega \approx 2.37$ is the current matrix multiplication constant and $\tilde O$ hides logarithmic factors. \changed{We elaborate on the fixed-point computation in \Cref{sec:rm}.} This improves substantially on the result of \citet{Farina23:Polynomial}, which has the same regret bound but whose per-iteration computation involved a {\em quadratic} program (namely, an $\ell_2$ projection), which has higher complexity than a linear program (they give a bound of $\tilde O(d^{10})$). Finally, we demonstrate via experiments that our method is also empirically faster than the prior method.
\section{Preliminaries}
Here, we review fundamentals of tree-form decision making, extensive-form games, and online convex optimization. Our exposition and notation mostly follows \citet{Farina23:Polynomial}.
\subsection{Tree-Form Decision Making}
A {\em tree-form decision problem} is a rooted tree where every path alternates between two types of nodes: {\em decision points} $(j \in \mc J)$ and {\em observation points} (or {\em sequences}) $(\sigma \in \Sigma)$. The root node $\Root \in \Sigma$ is always an observation point. At decision points, the edges are called {\em actions}, and the player must select one of the legal actions. At observation points, the edges are called {\em observations} or {\em signals}, and the player observes one of the signals before continuing. The number of sequences is denoted $d = |\Sigma|$. The parent of a node $s$ is denoted $p_s$. The set of actions available at a decision point $j$ is denoted $A_j$. The set of decision point following an observation point $\sigma$ will be denoted $C_\sigma$. An observation node $\sigma \in \Sigma$ is uniquely identified by its parent decision point $j$ and the action $a$ taken at $j$. We will hence use $ja$ as an alternative notation for the same observation point.

A {\em sequence-form pure strategy} for the player is a vector $\vec x \in \{0, 1\}^d$, indexed by sequences in $\Sigma$, where $\vec x(\sigma) = 1$ if the player selects {\em every} action on the $\Root \to \sigma$ path. A {\em sequence-form mixed strategy} is a convex combination of sequence-form pure strategies. We will use $\mc X$ to denote the set of sequence-form pure strategies. An important property~\cite{Romanovskii62:Reduction,Stengel96:Efficient} is that the convex hull of $\mc X$, which we will denote $\co\mc X$, is described by a system of linear constraints:
\begin{align}
    \vec x(\Root) = 1, \qq{} \vec x(p_j) = \sum_{a \in A_j} \vec x(ja) \ \forall j \in \mc J.\label{eq:sequence-form}
\end{align}
Tree-form decision problems naturally encode the decision problems faced by a player with perfect recall in an {\em extensive-form game}. An extensive-form game with $n$ players is a game of incomplete information played on a tree of nodes $\mc H$. At every non-leaf node $h \in H$, the children of $h$ are labeled with actions $a \in A_h$. Each nonterminal node is assigned to a different player, and the player to whom a node is assigned selects the action at that node. The nodes assigned to a given player are partitioned into {\em information sets}, or {\em infosets}; a player cannot distinguish among the nodes in a given infoset, and therefore a pure strategy must play the same action at every node in an infoset. Finally, each player has a {\em utility function} $u_i : \mc Z \to \R$, where $\mc Z$ is the set of terminal nodes. We will assume {\em perfect recall}, that is, we will assume that players never forget information.

In an extensive-form game, a perfect-recall player's decision problem is a tree-form decision problem whose size (number of nodes) is linear in the size of the game tree, and the utility functions are linear in every player's strategy. We will use $\mc X_i$ to denote the tree-form decision problem faced by player $i$. Then the utility functions $u_i : \co\mc X_1 \times \dots \times \co\mc X_n \to \R$ are linear in each player's strategy.

\subsection{Online Convex Optimization and $\Phi$-Regret}\label{sec:prelims-phi-regret}

In {\em online convex optimization}~\cite{Zinkevich03:Online}, a player (or ``learner'') has a strategy  set $\mc X \subseteq \R^d$, and repeatedly faces a possibly-adversarial environment. More formally, at every iteration $t = 1, \dots, T$, the player selects a distribution $\pi ^{(t)} \in \Delta(\mc X)$, and an adversary simultaneously selects a {\em utility vector} $\vec u^{(t)} \in [0,1]^d$. The player then observes the utility $\vec u^{(t)}$, selects a new strategy $\vec x^{(t+1)} \in \mc X$, and so on. Our metric of performance will be the notion of {\em $\Phi$-regret}~\cite{Greenwald03:General}. Given a set of transformations\footnote{$A^B$ is the set of functions from $B$ to $A$} $\Phi \subseteq (\co \mc X)^{\mc X}$:
\begin{definition}\label{def:phi-regret}
    The {\em $\Phi$-regret} of the player after $T$ timesteps is given by
    \begin{align}
        \Reg_{\Phi}(T) := \max_{\phi \in \Phi} \sum_{t=1}^T \E_{\vec x\sim \pi^{(t)}} \ip*{\vec u^{(t)},\phi(\vec x) - \vec x}.
    \end{align}
\end{definition}
Various choices of $\Phi$ correspond to various notions of regret, with larger sets resulting in stronger notions of regret. In an extensive-form game, notions of $\Phi$-regret correspond to notions of equilibrium. For each player $i \in [n]$, let $\Phi_i \subseteq (\co \mc X_i)^{\mc X_i}$ be a set of transformations for player $i$.

    A distribution $\pi \in \Delta(\mc X_1 \times \dots \times \mc X_n)$ is called a {\em correlated profile}. A $\eps$-$(\Phi_i)$-equilibrium is a correlated profile such that
    $        \E_{\vec x \sim \pi} \qty[u_i(\phi(\vec x_i), \vec x_{-i}) - u_i(\vec x_i, \vec x_{-i})] \le \eps.
    $
    for every player $i$ and deviation $\phi \in \Phi_i$.
If all players independently run $\Phi_i$-regret minimizers over their strategy sets $\mc X_i$, the empirical frequency of play $\pi = \op{Unif}(\pi^{(1)}, \dots, \pi^{(T)})$ will be an $\eps$-$(\Phi_i)$-equilibrium for $\eps = \max_i \Reg_{\Phi_i}(T)/T$. Thus, $\Phi$-regret minimizers immediately imply no-regret learning algorithms converging to $(\Phi_i)$-equilibria. \changed{Some common choices of $\Phi$, and corresponding equilibrium notions, are in \Cref{tab:eqconcepts}.}
\begin{table}
\centering
\scalebox{0.8}{
    \begin{tabular}{lll}
        \toprule
        \bf Deviations $\Phi$ & \bf Equilibrium concept & \bf References \\
        \midrule
        Constant (external), $\Phi = \{\phi: \vec x \mapsto \vec x_0 \mid \vec x_0 \in \mc X\}$ & Normal-form coarse correlated & \citet{Moulin78:Strategically} \\
        Trigger (see \Cref{sec:mediator}) & Extensive-form correlated & \citet{Stengel08:Extensive} \\
        Communication (see \Cref{sec:mediator}) & Communication & \citet{Forges86:Approach,Myerson86:Multistage}\\
        Linear / Untimed communication & Linear correlated & \citet{Farina23:Polynomial}; {\bf this paper}\\
        Swap, $\Phi = \mc X^{\mc X}$ & Normal-form correlated & \citet{Aumann74:Subjectivity} \\
        \bottomrule
    \end{tabular}
}
\caption{\changed{Some examples of deviation sets $\Phi$ and corresponding notions of correlated equilibrium,  in increasing order of size of $\Phi$ (and thus increasing tightness of the equilibrium concept)}}\label{tab:eqconcepts}
\end{table}
In this paper, our focus will be on {\em linear-swap regret}, which is the regret against the set $\dPhi{Lin}$ of all linear\footnote{For sets $\mc X$ whose affine hull excludes the origin, there is no point in distinguishing affine maps from linear maps. Sequence-form strategy sets $\mc X$ are such sets, because $\vec x(\Root) = 1$ is always a constraint. So, throughout this paper, we will not distinguish between linear and affine maps.} maps $\phi : \mc X \to \co \mc X$. To our knowledge, linear-swap regret was first proposed by \citet{Gordon08:No} for general convex spaces. They developed a general framework for $\Phi$-regret minimization, which we now review. We start by observing that any linear $\phi : \mc X \to \co \mc X$ is naturally extended to a function $\phi : \co \mc X \to \co \mc X$ by setting $\phi: \co \mc X \ni \vec x \mapsto \E_{\vec x' \sim \pi} \phi(\vec x')$, where $\pi \in \Delta(\mc X)$ is any distribution for which $\E_{\vec x' \sim \pi} \vec x' = \vec x$ (The choice of distribution is irrelevant because of linearity of expectation, and thus $\phi : \co \mc X \to \co \mc X$ is uniquely defined.)

\begin{theorem}[\citealp{Gordon08:No}]\label{th:gordon}
    Let $\Phi \subseteq \dPhi{Lin}$ be a convex set of transformations, and let  $\mc R_\Phi$ be a deterministic\footnote{A deterministic regret minimizer is one that uses no randomness internally to compute its strategies. When the strategy set (here $\Phi$) is convex, and the notion is external regret, the learner need not randomize: since utilities are linear, picking a distribution $\pi^{(t)} \in \Delta(\Phi)$ is equivalent to deterministically selecting the point $\phi^{(t)} := \E_{\phi\sim \pi^{(t)}} \phi$. Thus, we allow the $\mc R_\Phi$-adversary (here, the player itself) to set a utility $\phi \mapsto \ip*{\vec u^{(t)}, \phi(\vec x^{(t)})}$ that depends on the learner's choice of $\phi$.} external regret minimizer over $\Phi$,  whose regret after $T$ timesteps is $R$. Then the following algorithm achieves $\Phi$-regret $R$ on $\mc X$ after $T$ timesteps: At every timestep $t$, the player queries $\mc R_\Phi$ for a strategy (transformation) $\phi^{(t)} \in \Phi$, and the player selects a strategy $\vec x^{(t)} \in \co \mc X$ that is a fixed point of $\phi^{(t)}$, that is, $\phi^{(t)}(\vec x^{(t)}) = \vec x^{(t)}$. Upon observing utility $\vec u^{(t)}$, the player forwards the utility $\phi \mapsto \ip*{\vec u^{(t)}, \phi(\vec x^{(t)})}$ to $\mc R_\Phi$.
\end{theorem}
Therefore, to construct a $\dPhi{Lin}$-regret minimizer over $\mc X$, it suffices to be able to (1) minimize {\em external} regret over $\dPhi{Lin}$, and (2) compute fixed points of transformations $\phi^{(t)}$. For linear $\phi^{(t)}: \vec x \mapsto \mat A \vec x$, computing a fixed point amounts to solving a linear program. Therefore, the focus of this paper will be on external regret minimizers over the set $\dPhi{Lin}$.

For extensive-form games, linear-swap regret was recently studied in detail by \citet{Farina23:Polynomial}: they provide a characterization of the set $\dPhi{Lin}$ when $\mc X$ is a sequence-form polytope, and thus derive an algorithm for minimizing $\dPhi{Lin}$-regret over $\mc X$. Their paper  is the starting point of ours.

\section{Mediators and UTC Deviations}\label{sec:mediator}
With the notable exception  of linear deviations, most sets of deviations $\Phi$ for extensive-form games are defined by interactions between a {\em mediator} who holds a strategy $\vec x \in \mc X$, and a {\em deviator}, who should compute the function $\phi(\vec x)$ by making queries to the mediator. The set of deviations is then defined by what queries that the player is allowed to make. Before continuing, we will first formulate the sets $\Phi$ mentioned in \Cref{sec:prelims-phi-regret} in this paradigm, for intuition. For a given decision point $j$, call an action $a \in A_j$ the {\em recommended action at $j$}, denoted $a(\vec x, j)$, if $\vec x(ja) = 1$. Since $\vec x$ is a sequence-form strategy, it is possible for a decision point to have no recommended action if its parent $p_j$ is itself not recommended.
\begin{itemize}
    \item Constant (NFCCE): The deviator cannot to make any queries to the mediator.
    \item Trigger (EFCE): The deviator, upon reaching a decision point $j$, learns the recommended action (if any) at $j$ before selecting its own action.
    \item Communication: The deviator maintains a {\em state} with the mediator, which is a sequence $\sigma$, initially $\Root$. Upon reaching a decision point $j$, the deviator selects a decision point $j' \in C_\sigma$ (possibly $j' \ne j$) at which to query the mediator, the deviator observes the recommendation $a' = a(\vec x, j')$, then the deviator must pick an action $a \in A_j$. \changed{The state is updated to $j'a'$.}
    \item Swap (NFCE): The deviator learns the whole strategy $\vec x$ before selecting its strategy.
\end{itemize}
\changed{An example of a communication deviation can be found in \Cref{sec:example}, and further discussion of these solution concepts can be found in \Cref{sec:related}.}
Of these, the closest notion to ours is the notion of communication deviation, and that is the starting point of our construction. One critical property of communication deviations is that the mediator and deviator ``share a clock'': for every decision point reached, the deviator must make exactly one query to the mediator. As the name suggests, our set of {\em untimed} deviations results from removing this timing restriction, and therefore allowing the deviator to make {\em any number} (zero, one, or more than one) of queries to the mediator for every decision point reached. We formally define the decision problem faced by an untimed deviator as follows.
\begin{definition}\label{def:utc}
    The {\em UTC decision problem} corresponding to a given tree-form decision problem is defined as follows. Nodes are identified with pairs $(s, \tilde s)$ where $s, \tilde s \in \Sigma \cup \mc J$. $s$ represents the state of the real decision problem, and $
        \tilde s$ represents the state of the mediator. The root is $(\Root, \Root) \in \Sigma \times \Sigma$.
    \begin{enumerate}
        \item $(\sigma, \tilde\sigma) \in \Sigma \times \Sigma$ is an observation point.  The deviator observes the next decision point $j \in C_\sigma$, and the resulting decision point is $(j, \tilde\sigma)$
        \item $(j, \tilde\jmath) \in \mc J \times \mc J$ is an observation point. The deviator observes the recommendation $a = a(\vec x, \tilde\jmath)$, and the resulting decision point is $(j, \tilde\jmath a)$.
        \item $(j, \tilde\sigma) \in \mc J \times \Sigma$ is a decision point. The deviator can choose to either play an action $a \in A_j$, or to query a decision point $\tilde\jmath \in C_{\tilde\sigma}$. In the former case, the resulting observation point is $(ja, \tilde\sigma)$ for $a \in A_j$; in the latter case, the resulting observation point is $(j, \tilde\jmath)$.
    \end{enumerate}
\end{definition}

Any mixed strategy of the deviator in this decision problem defines a function $\phi : \mc X \to \co\mc X$, where $\phi(\vec x)(\sigma)$ is the probability that an untimed deviator plays all the actions on the path to $\sigma$ when the mediator recommends according to pure strategy $\vec x$. We thus define:
\begin{definition}
    An {\em UTC deviation} is any function $\phi : \mc X \to \co\mc X$ induced by a mixed strategy of the deviator in the UTC decision problem.
\end{definition}
Clearly, the set of UTC deviations is at least as large as the set of communication deviations, and at most as large as the set of swap deviations. In the next section, we will discuss how to represent UTC deviations, and show that UTC deviations coincide precisely with linear deviations.

\section{Representation of UTC Deviations and Equivalence between UTC and Linear Deviations}
Since UTC deviations are defined by a decision problem, one method of representing such deviations is to express it as a tree-form decision problem and use the sequence-form representation. However, the UTC decision problem is not a tree---it is a DAG, since there are multiple ways of reaching any given decision point $(j, \tilde\sigma)$ depending on the ordering of the player's past actions and queries. Converting it to a tree by considering the tree of paths through the DAG would result in an exponential blowup: a decision point $(j, \tilde\sigma)$, where $j$ is at depth $k$ and $\tilde\sigma$ is at depth $\ell$, can be reached in roughly $\binom{k + \ell}{k}$ ways, so the total number of paths can be exponential in the depth of the decision problem even when the number of sequences, $d = |\Sigma|$, is not.

However, it is still possible to define the ``sequence form'' of a pure deviation in our UTC decision problem as follows\footnote{This construction is a special case of the more general construction of sequence forms for DAG decision problems explored by \citet{Zhang23:Team_DAG} in the case of team games.}: it is a pair of matrices $(\mat A, \mat B)$ where $\mat A \in \{0, 1\}^{\Sigma\times\Sigma}$ encodes the part corresponding to sequences $(\sigma, \tilde\sigma)$, and $\mat B \in \{0,1\}^{\mc J \times \mc J}$ encodes the part corresponding to decision points $(j, \tilde\jmath)$. $\mat A(\sigma, \tilde\sigma) = 1$ if the deviator plays all the actions on {\em some} path to observation point $(\sigma, \tilde\sigma)$, and similarly  $\mat B(j, \tilde\jmath) = 1$ if the deviator plays all the actions on some path to observation node $(j, \tilde\jmath)$. Since the only possible way for two paths to end at the same observation point is for the deviator to have changed the order of actions and queries, for any given pure strategy of the deviator, at most one path can exist for both cases. Therefore, the set of mixed sequence-form deviations can be expressed using the following set of constraints:
\begin{align}
    \begin{aligned}
        \mat A(p_j, \tilde\sigma) + \mat B(j, p_{\tilde\sigma}) & = \sum_{a \in A_j} \mat A(ja, \tilde\sigma) + \sum_{\tilde\jmath \in C_{\tilde\sigma}} \mat B(j, \tilde\jmath) &  & \qq{$\forall j \in \mc J, \tilde\sigma \in \Sigma$} \\
        \mat A(\Root, \Root)                                    & = 1                                                                                                                                                                     \\
        \mat A(\Root, \tilde\sigma)                             & = 0                                                                                                            &  & \qq{$\forall \tilde\sigma \ne \Root$}
        \\
        \mat A, \mat B                                          & \ge 0
    \end{aligned}\label{eq:constraint system}
\end{align}
where, in a slight abuse of notation, we define $\mat B(j, p_\Root) := 0$ for every $j \in \mc J$.
Moreover, for any pair of matrices $(\mat A, \mat B)$ satisfying the constraint system and therefore defining some deviation $\phi : \mc X \to \co \mc X$, it is easy to compute how $\phi$ acts on any $\vec x \in \mc X$: the probability that the deviator plays all the actions on the $\Root\to\sigma$ path is simply given by
\begin{align}
    \sum_{\tilde\sigma \in \Sigma} \vec x(\tilde\sigma) \mat A(\sigma, \tilde\sigma) = (\mat A \vec x)(\sigma),
\end{align}
and therefore $\phi$ is nothing more than a matrix multiplication with $\mat A$, that is, $\phi(\vec x) = \mat A \vec x$. We have thus shown that every UTC deviation is linear, that is, $\dPhi{UTC} \subseteq \dPhi{Lin}$. In fact, the reverse inclusion holds too:
\begin{restatable}{theorem}{thMain}\label{th:main}
    The UTC deviations are precisely the linear deviations. That is,  $\dPhi{UTC} = \dPhi{Lin}$.
\end{restatable}
The proof is deferred to \Cref{app:proofs}. Since the two sets are equivalent, in the remainder of the paper, we will use the terms {\em UTC deviation} and {\em linear deviation} (similarly, {\em UTC regret} and {\em linear-swap regret}) interchangeably.

\section{Example}\label{sec:example}
\input{figures/example}
\newcommand{\profile}[2]{(\action1a#1, \action1b#2, \action1c1, \action2f#1, \action2g#2)}
In this section, we provide an example in which the UTC deviations are strictly more expressive than the communication deviations. Consider the game in \Cref{fig:example}. The subgames rooted at \util1{D} and \util1E are guessing games, where \pone must guess \ptwo's action, with a large penalty for guessing wrong. Consider the correlated profile that mixes uniformly among the four pure profiles \profile{i}{j} for $i, j \in \{1, 2\}$. In this profile, the information that \pone needs to guess perfectly is contained in the recommendations: the recommendation at \util1{A} tells it how to guess at \util1D, and the recommendation at \util1B tells it how to guess at \util1E. With a communication deviation, \pone cannot access this information in a profitable way, since upon reaching \util1C, \pone must immediately make its first mediator query. Hence, this profile is a communication equilibrium. However, with an {\em untimed} communication deviation, \pone can profit: it should, upon reaching\footnote{The actions/queries \pone makes at \util1A and \util1B are irrelevant, because \changed{\pone only cares about maximizing utility, and} it always gets utility $0$ regardless of what it does. In the depiction of this deviation in \Cref{fig:example-dev}, the deviator always plays action 1 at \util1A and \util1B.} \util1C, play action \action1c2 {\em without making a mediator query}, and then query \util1A if it observes \util1D, and \util1B if it observes \util1E. \changed{This deviation is allowed only due to the untimed nature of UTC deviations allows the deviating player to {\em delay} its query to the mediator until it reaches either \util1D or \util1E. In a {\em timed} communication deviation, this deviation is impossible, because the player must make its first query (\util1A, \util1B, or \util1C) {\em before} reaching \util1D or \util1E, and thus that query cannot be conditioned on which one of \util1D or \util1E will be reached. }

Another example, where the player can profit from making {\em more than one} query, and untimed deviations affects the set of possible equilibrium outcomes, can be found in \Cref{sec:more-examples}.

\section{Regret Minimization on $\dPhi{UTC}$}\label{sec:rm}
\begin{figure}[t] %
    \resizebox{.5\textwidth}{!}{\import{plots}{K45.pgf}}%
    \resizebox{.5\textwidth}{!}{\import{plots}{U222.pgf}}\\[-4mm]
    \resizebox{.5\textwidth}{!}{\import{plots}{L3132.pgf}}%
    \resizebox{.5\textwidth}{!}{\import{plots}{S2123.pgf}}\\[-1cm]
    \caption{Experimental comparison between our dynamics and those of \citet{Farina23:Polynomial} for approximating a linear correlated equilibrium in extensive-form games. Each algorithm was run for a maximum of $100,\!000$ iterations or 6 hours, whichever was hit first. Runs that were terminated due to the time limit are marked with a square \small{$\blacksquare$}.}
    \label{fig:experiments}
\end{figure}
\label{sec:regret}

In this section, we discuss how \Cref{th:main} can be used to construct very efficient $\dPhi{Lin}$-regret minimizers, both in theory and in practice. The key observation we use here is due to \citet{Zhang23:Team_DAG}: they observed that DAG decision problems have a structure that allows them to be expressed as {\em scaled extensions}, allowing the application of the {\em counterfactual regret minimization} (CFR) framework~\cite{Zinkevich07:Regret,Farina19:Regret}:
\begin{theorem}[CFR for $\dPhi{Lin}$, special case of \citealp{Zhang23:Team_DAG}]\label{th:cfr}
    \changed{CFR-based algorithms can be used to construct an external regret minimizer on $\dPhi{UTC}$ (and thus also on $\dPhi{Lin}$)} with $O(d^2 \sqrt{T})$ regret \changed{and $O(d^2)$ per-iteration complexity.} 
\end{theorem}%
\changed{Applying \Cref{th:gordon} now yields:
\begin{theorem}
    CFR-based algorithms can be used to construct a $\dPhi{Lin}$-regret minimizer with $O(d^2 \sqrt{T})$ regret, and per-iteration complexity dominated by the complexity of computing a fixed point of a linear transformation $\phi^{(t)} :\co \mc X \to \co\mc X$.
\end{theorem}
}

As mentioned in the introduction, this significantly improves the per-iteration complexity of linear-swap regret minimization\changed{. Fixed points can be computed by finding a feasible solution to the constraint system $\{ \vec x \in \mc X, \mat A \vec x = \vec x\} $, where $\vec x \in \mc X$ is expressed using the sequence-form constraints~\eqref{eq:sequence-form}. This is a linear program with $O(d)$ variables and constraints, so the LP algorithm of \citet{Cohen21:Solving} yields a fixed-point computation algorithm with runtime $\tilde O(d^\omega)$.}

\changed{For comparison, the algorithm of \citet{Farina23:Polynomial} requires an $\ell_2$ projection onto $\mc X$ on every iteration, which requires solving a convex quadratic program; the authors of that paper derive a bound of $\tilde O(d^{10})$, which, although polynomial, is much slower than our algorithm.} CFR-based algorithms are currently the fastest practical regret minimizers~\cite{Brown19:Solving,Farina21:Connecting}---therefore, showing that our method allows such algorithms to be applied is also a significant practical step. In \Cref{sec:experiments}, we will show empirically that the resulting algorithm is significantly better than the previously-known state of the art, in terms of both per-iteration time complexity and number of iterations.

\section{Experimental Evaluation}\label{sec:experiments}

We empirically investigate the performance of our learning dynamics for linear correlated equilibrium, compared to the recent algorithm by \citet{Farina23:Polynomial}.
We test on four benchmark games:
\begin{itemize}[left=5mm]
    \item 4-player Kuhn poker, a multiplayer variant of the classic benchmark game introduced by \citet{Kuhn50:Simplified}. The deck has 5 cards. This game has $3,\!960$ terminal states.
    \item A ridesharing game, a two-player general-sum game introduced as a benchmark for welfare-maximizing equilibria by \citet{Zhang22:Optimal}. This game has $484$ terminal states.
    \item 3-player Leduc poker, a three-player variant of the classic Leduc poker introduced by \citet{Southey05:Bayes}. Only one bet per round is allowed, and the deck has 6 cards (3 ranks, 2 suits). The game has $4,\!500$ terminal states.
    \item Sheriff of Nottingham, a two-player general-sum game introduced by \citet{Farina19:Correlation} for its richness of equilibrium points. The smuggler has 10 items, a maxmimum bribe of 2, and 2 rounds to bargain. The game has $2,\!376$ terminal states.
\end{itemize}

We run our algorithm based on the UTC polytope, and that of \citet{Farina23:Polynomial} (with the learning rate $\eta=0.1$ as used by the authors), for a limit of $100,\!000$ iterations or 6 hours, whichever is hit first. \changed{Instead of solving linear programs to find the fixed points, we use power iteration, which is faster in practice.}
All experiments were run on the same machine with 32GB of RAM and a processor running at a nominal speed of 2.4GHz. For our learning dynamics, we employed the CFR algorithm instantiated with the regret matching$^+$ \citep{Tammelin14:Solving} regret minimizer at each decision point (see \Cref{th:cfr}).
Experimental results are shown in \Cref{fig:experiments}.

\changed{One of the most appealing features of our algorithm is that allows CFR-based methods to apply. CFR-based methods are the fastest regret minimizers in practice, so it is unsurprising that using them results in better convergence as seen in \Cref{fig:experiments}. Another appealing feature} is that our method sidesteps the need of projecting onto the set of transformations.
This is in contrast with the algorithm of \citet{Farina23:Polynomial}, which requires an expensive projection at every iteration. We observe that this difference results in
a dramatic reduction in iteration runtime between the two algorithms, which we quantify in \Cref{tab:iteration times}.
So, we remark that when accounting for \emph{time} instead of iterations on the x-axis of the plots in \Cref{fig:experiments}, the difference in performance between the algorithms appears even stronger. Such a plot is available in \Cref{app:plots}.

\begin{table}[t]
    \setlength{\tabcolsep}{1mm}
    \centering\begin{tabular}{llrcrrcrr}
        \toprule
        \bf Game              & \qquad\qquad & \multicolumn{3}{r}{\bf Our algorithm} & \multicolumn{3}{r}{\bf \citet{Farina23:Polynomial}}     & \quad \bf Speedup                                         \\
        \midrule
        4-Player Kuhn poker   &              & 5.65ms                                 & $\pm$                                               & 0.30ms & \qquad\qquad 195ms & $\pm$ & 7ms  & 35$\times$ \\
        Ridesharing game      &              & 676\textmu s                                 & $\pm$                                               & 80\textmu s & 160ms              & $\pm$ & 7ms  & 237$\times$ \\
        3-Player Leduc poker  &              & 42.0ms                                  & $\pm$                                               & 0.7ms & 12.1s            & $\pm$ & 1.0s & 287$\times$ \\
        Sheriff of Nottingham &              & 114ms                                 & $\pm$                                               & 16ms  & 50.2s              & $\pm$ & 9.6s & 442$\times$ \\
        \bottomrule
    \end{tabular}
    \caption{Comparison of average time per iteration. For each combination of game instance and algorithm, the mean and standard deviation of the iteration runtime are noted.\\}
    \label{tab:iteration times}
    
    \setlength{\tabcolsep}{1mm}
    \centering\begin{tabular}{lrrrrrrrrrrr}
        \toprule
        \bf Game              &\quad\bf  Target gap & \quad{\bf Our algorithm} & \quad{\bf \citet{Farina23:Polynomial}}     & \quad \bf Speedup                                         \\
        \midrule
        4-Player Kuhn poker   &     $7 \times 10^{-4}$       &  32.8s   &  5h 25m  & 595$\times$ \\
        Ridesharing game      &  $9 \times 10^{-5}$ & 8.89s & 4h 07m & 1667$\times$ \\
        3-Player Leduc poker  & $0.224$ & 2.12s & 6h 00m & 10179$\times$ \\
        Sheriff of Nottingham &  $2.06$ & 2.00s & 6h 00m & 
        10800$\times$ \\
        \bottomrule
    \end{tabular}
    \caption{Comparison of time taken to achieve a particular linear swap equilibrium gap. The gap is whatever gap was achieved by the algorithm of \citet{Farina23:Polynomial} before termination.}
    \label{tab:conv}
\end{table}

\section{Conclusion and Future Research}

In this paper, we have introduced a new representation for the set of linear deviations when the strategy space is sequence form. Our representation connects linear deviations to the mediator-based framework that is more typically used for correlation concepts in extensive-form games, and therefore gives a reasonable game-theoretic interpretation of what linear equilibria represent. It also leads to state-of-the-art no-linear-regret algorithms, both in theory and in practice. Several natural questions remain open:
\begin{enumerate}[left=5mm]
    \item \changed{Is there an algorithm whose swap regret is $\poly(d) \cdot T^c$ for $c < 1$ in extensive-form games? (See also \Cref{sec:related} for some recent progress on this problem.)}
    \item What would be a reasonable definition of {\em untimed communication equilibrium}, as a refinement of {\em communication equilibrium} (see also \Cref{sec:revelation})?
    \item For extensive-form correlated equilibrium, it is possible to achieve $\poly(d) \cdot \log(T)$ regret~\cite{Anagnostides23:NearOptimal}, and to compute exact equilibria in polynomial time~\cite{Huang08:Computing}. Can one extend these results to linear equilibria?
\end{enumerate}

\section*{Acknowledgements}
This material is based on work supported by the Vannevar Bush Faculty
Fellowship ONR N00014-23-1-2876, National Science Foundation grants
RI-2312342 and RI-1901403, ARO award W911NF2210266, and NIH award
A240108S001.
\bibliographystyle{iclr2024_conference}
\bibliography{dairefs}

\newpage
\appendix

\section{Discussion}

Here, we discuss a few points of interest about UTC and linear deviations.

\subsection{The Convex Hull of Pure Deviations}\label{sec:convhull}

In our definitions, we were careful to allow transformations $\phi \in \Phi$ to map from the set of pure strategies, $\mc X$, to its convex hull $\co \mc X$, instead of insisting that every pure strategy map onto another pure strategy. One might ask whether this makes a difference in our definitions. For example, if in \Cref{def:phi-regret} we restrict our attention to $\phi : \mc X \to \mc X$, does the definition change? In symbols, for a given set of transformations $\Phi$, is $\Phi \subseteq \co {\hat \Phi}$ where $\hat \Phi = \{ \phi \in \Phi : \phi(\mc X) \in \mc X\ \forall \vec x \in \mc X\}$? For the other sets of deviations mentioned in this paper (external, swap, trigger, and communication), the answer is already known to be positive.

Our equivalence theorem between UTC and linear deviations gives an answer to this question for the set of linear deviations as well. Since the UTC deviations are defined by a decision problem, every mixed UTC deviation is by definition equivalent to a distribution over pure UTC deviations. That is, the vertices of $\dPhi{UTC}$ are the pure UTC deviations: they map pure strategies $\mc X$ to pure strategies. Since $\dPhi{UTC} = \dPhi{Lin}$, this proves:
\begin{corollary}\label{cor:convhull}
    When $\mc X$ is a sequence-form polytope, the extreme points of $\dPhi{Lin}$ are the linear maps $\phi : \mc X \to \mc X$, \ie, the linear maps that map pure strategies to pure strategies. Thus, $\dPhi{Lin} = \hat\Phi_\textsc{Lin}$
\end{corollary}
This result is not obvious {\em a priori}.  For example, it fails to generalize to other sets of functions $\Phi$, or to $\dPhi{Lin}$ for polytopes $\mc X$ that are not sequence-form polytopes:
\begin{itemize}
    \item Other sets of functions $\Phi$: Let $\Phi = \{ \phi \}$ consist of a single constant function $\phi : \vec x \mapsto \vec x^*$, where $\vec x^* \in (\co \mc X) \setminus \mc X$. Then $\hat\Phi$ is empty, so $\Phi \not\subseteq \co \hat\Phi$.
    \item Non-sequence-form polytopes: Take $\mc X$ to be a trapezoid $ABCD$ where $AB$ is the longer of the two bases and consider the linear map $\phi$ with $\phi(A) = D, \phi(D) = A,$ and $\phi(B) = C$. Then $\phi$ is an extreme point of $\dPhi{Lin}$, but $\phi(C)$ will lie somewhere along segment $AB$, but at neither endpoint---that is, not at a vertex. \Cref{fig:trapezoid} has a visual depiction.
\end{itemize}

\input{figures/trapezoid}
\subsection{Generalization to Arbitrary Pairs of Polytopes}\label{sec:generalization}
Our main result characterizes the set of linear maps $\phi : \mc X \to \mc X$ for sequence-form polytopes $\mc X$. However, it is actually more general than this: an identical proof works to characterize the set of linear maps $\phi : \mc Y \to \co \mc X$ for any (possibly different!) sequence-form polytopes $\mc X$ and $\mc Y$. Hence, we have shown:
\begin{theorem}\label{th:general pair}
    Let $\mc X, \mc Y$ be sequence-form strategy sets. The linear maps $\phi : \mc Y \to \co \mc X$ are precisely the functions induced by strategies in the DAG decision problem whose nodes are identified with pairs $(s, \tilde s)$, where $s \in \Sigma_{\mc X} \cup \mc J_{\mc X}$ and $\tilde s \in \Sigma_{\mc Y} \cup \mc J_{\mc Y}$, and which behaves analogously to \Cref{def:utc}.
\end{theorem}
Although we are mostly concerned with the case $\mc X = \mc Y$ in this paper, we state this extension in the hope that it may be of independent interest. We will also use it in the proof of the revelation principle (\Cref{th:rp}).

\subsection{Uniqueness of Representation}

The statement of \Cref{th:main} discusses $\dPhi{UTC}$ and $\dPhi{Lin}$ as {\em sets of functions} $\Phi \subseteq (\co\mc X)^{\mc X}$. It does {\em not} imply that for every linear map $\phi : \mc X \to \co\mc X$ there is {\em exactly} one representation of $\phi$ as a deviator strategy in the UTC decision problem, only that there is {\em at least} one representation. Indeed, the external deviations (constant functions $\phi : \vec x \mapsto \vec x^*$ for fixed $\vec x^* \in \co\mc X$ can be represented via a large number of different strategies in the UTC decision problem: the deviator may send any number of queries to the mediator, before eventually deciding to ignore the queries and play according to $\vec x^*$, and such a deviator would still represent the external deviation $\phi$.

Similarly, \Cref{th:main} also does not imply that every matrix $\mat A \in \R^{\Sigma \times \Sigma}$ representing a linear map $\phi_{\mat A} : \mc X \to \co\mc X$ is part of a pair $(\mat A, \mat B)$ satisfying the system of equations \eqref{eq:constraint system}. Indeed, the proof of \Cref{th:main} only shows that for every linear $\phi : \mc X \to\co\mc X$, there is {\em at least one} pair $(\mat A, \mat B)$ satisfying \eqref{eq:constraint system} where $\mat A$ represents $\phi$. It is easy to construct matrices $\mat A$ that represent linear maps, yet cannot satisfy \eqref{eq:constraint system}, by changing the first row of $\mat A$ to some other vector $\vec c$ with $\vec c^\top \vec x = 1$ for all $\vec x \in \mc X$.

\subsection{When Untimed and Timed Communication Deviations Coincide}

If all players have only one layer of decision nodes, the game is a single-stage {\em Bayesian game}---in that special case, the communication deviations and linear deviations will coincide\footnote{Here, by two sets of deviations {\em coinciding}, we mean that the same set of deviation functions $\phi : \mc X \to \co \mc X$ is available to both deviators.}. This property was also proven by \citet{Fujii23:Bayes}, but our framework gives a particularly simple proof via \Cref{th:main}: for any UTC deviation in a single-stage game, the deviator makes either no queries or one query to the mediator. A communication deviator can simulate the same function by making the same query (if any), or, if the UTC deviator makes no query, by making an arbitrary query and ignoring the reply. It turns out that the converse is also essentially true:
\begin{theorem}
    Consider any decision problem with no nontrivial decision points---that is, the player has at least two legal actions at every decision point. The (timed) communication deviations coincide with the untimed communication deviations (and hence also the linear deviations) {\em if and only if} every path through the decision problem contains at most one decision point.
\end{theorem}

\begin{proof}
    The {\em if} direction was shown above and by \citet{Fujii23:Bayes}, so it suffices to show the {\em only if} direction. Suppose there are two decision points, $A$ and $B$, such that $B$ is a child of action $a_1$ at $A$. Let $a_2$ be another action at $A$, and let $b_1$ and $b_2$ be two actions at $B$. (The game in \Cref{fig:example2} has such a structure). Consider any deviation $\phi$ that plays action $a_i$ if it is recommended $b_i$, for $i \in \{ 1, 2 \}$. It is easy to construct untimed deviations with this behavior, but timed deviations cannot have this behavior, because a timed deviation cannot know the recommendation at $B$ while still at decision point $A$.
\end{proof}

\subsection{Relation between Our Representation and That of \citet{Farina23:Polynomial}}
Our paper and the paper of \citet{Farina23:Polynomial} both take similar approaches to minimizing $\dPhi{Lin}$-regret: both papers use the framework of \citet{Gordon08:No} to reduce the problem to minimizing external regret over the set of linear maps, and then derive a system of constraints for the set of matrices $\mat A \in \R^{\Sigma \times \Sigma}$  that represent linear maps. The representations are, however, significantly different:
\begin{itemize}
    \item The representation of \citet{Farina23:Polynomial} cannot be expressed in scaled extensions. As such, that paper was forced to resort to less efficient regret minimization techniques. This difference is what allows us to improve upon their results.
    \item As a technical note, the representation of \citet{Farina23:Polynomial} will always result in a matrix $\mat A \in \R^{\Sigma \times \Sigma}$ where the columns of $\mat A$ corresponding to nonterminal sequences are filled with zeros. While this is without loss of generality due to the constraints defining the sequence form, it sometimes results in intuitively-strange representations: for example, their representation does not represent the identity map $\op{Id} : \mc X \to \mc X$ as the identity matrix $\mat I \in \R^{\Sigma \times \Sigma}$, whereas our representation will.
    \item While our representation generalizes to arbitrary pairs of sequence-form polytopes according to \Cref{sec:generalization}, theirs generalizes even further, to functions $\phi : \mc Y \to \mc X$ as long as $\mc Y$ is sequence form and $\mc X$ has {\em some} small set of linear constraints (not necessarily sequence form) describing it. We likely cannot hope for our representation to generalize as far: our proof of equivalence relies fundamentally on both input and output being sequence-form strategy sets.
\end{itemize}

\subsection{Untimed Communication Equilibria}\label{sec:revelation}
The UTC deviations, like all sets of deviations, give rise to a notion of equilibrium. We define:
\begin{definition}
    In an extensive-form game, an {\em untimed private communication equilibrium} is a correlated profile that is a $(\Phi_i)$-equilibrium where $\Phi_i$ is player $i$'s set of UTC deviations.
\end{definition}

We add the word ``private'' here in the name to emphasize the fact that the mediator must have a separate interaction with each player---that is, the mediator cannot use its interactions with one player to inform how it gives recommendations to another player. This is enforced by the fact that the equilibrium is a correlated profile. See \Cref{footnote:comm} regarding why this distinction is important. 

Defining untimed communication equilibrium without such a privacy restriction seems to be a subtle task, and is orthogonal to and beyond the scope of the present work. \changed{However, we will make a few informal comments here. Untimed communication equilibria (without the privacy constraint) are difficult to define in a way that does not quickly collapse to the regular notion of communication equilibrium. In games with three or more players, the mediator is always guaranteed that two of the players have not deviated, and those two players will have messages synchronized with the game clock. Therefore, under reasonable assumptions on how often each player makes moves, the mediator will immediately know if the deviating player is sending out-of-order messages, and this concept would reduce immediately to the regular communication equilibrium. It is entirely unclear how to define a notion of untimed (non-private) communication equilibrium that does not exhibit such a collapse.}

\changed{In two-player games, it is possible that there is a reasonable way to define untimed communication equilibria. The above collapse does not apply, because the mediator will not know which player is the one sending out-of-timing messages. However, this definition would still be rather subtle---for example, when do the out-of-order messages arrive to the mediator, relative to the {\em other player's} messages? We leave these issues to future work.}

\subsection{Revelation Principle for Untimed Private Commnunication Equilibrium}
All the other notions of equilibrium involving a mediator, discussed in \Cref{sec:mediator}, obey a {\em revelation principle}, which we now discuss using the example of normal-form correlated equilibrium. The original definition of \citet{Aumann74:Subjectivity} did not initially refer to correlated profiles; instead, the definition posited an arbitrary joint distribution of correlated signals. In this section, we break down the notions of equilibrium that we have defined so far, and reconstruct them from the perspective of this arbitrary set of signals, and show that the resulting notions are equivalent for our notion of untimed private communication equilibrium.

We start with NFCE as an illustrative example. Let $\pi \in \Delta(\mc S_1 \times \dots \times \mc S_n)$, where $\mc S_i$ is an arbitrary set of signals for player $i$. The mediator samples a joint signal $(s_1, \dots, s_n) \sim \pi$, and then each player privately observes its own signal $s_i$ and selects (possibly at random) a strategy $\vec x \in \mc X$. An NFCE is then a tuple $(\pi, \phi_1, \dots, \phi_n)$, where $\phi_i : \mc S_i \to \co \mc X$ is the function by which player $i$ selects its strategy given a signal, such that each player's choice of $\phi_i$ maximizes that player's utility given $\pi$ and the other $\phi_i$s, that is,
\begin{align}
    \E_{(s_1, \dots, s_n) \sim \pi}[ u_i(\phi_i'(s_i), \phi_{-i}(s_{-i}) - u_i(\phi_i(s_i), \phi_{-i}(s_{-i})] \le 0
\end{align}
for every player $i$ and other possible function $\phi_i' : \mc S_i \to \co \mc X$. An NFCE is {\em direct} if $\mc S_i = \mc X_i$ and $\phi_i : \mc X_i \to \co \mc X_i$ is the identity function.
The {\em revelation principle} states that every NFCE is outcome-equivalent\footnote{By {\em outcome-equivalent}, we mean that the distribution over terminal nodes in the extensive-form game is the same in both equilibria.} to a direct equilibrium.

To generalize this to extensive-form and communication equilibria (timed and untimed), we follow the approach of \citet{Myerson86:Multistage,Forges86:Approach}. In their approach, a mediator of player $i$ is a map $M_i : \mc S_i^{\le H} \to \mc S_i$ (where $\mc S_i^{\le H}$ is the set of all sequences over $\mc S_i$ of length $\le H$, and $H$ is some large but finite number (at least the depth of player $i$'s decision problem.) that determines what message the mediator sends in reply to a player whose message history with the mediator is a finite sequence $\vec s = \{ s_i \}$. We will assume that $\mc S_i$ at least is expressive enough to send an empty message, a decision point, or an observation point: $\mc S_i \supseteq \mc J_i \sqcup \Sigma_i \sqcup \{ \bot \}$. The three notions of extensive-form correlated equilibrium, private communication equilibrium, and untimed private communication equilibrium will differ in how the player interacts with the mediator. We will describe player $i$'s interactions by a set of functions $\Phi_i  \subseteq (\co \mc X_i)^{\mc M_i}$ where $\mc M_i$ is a set of mediators: each function $\phi_i \in \Phi_i$ represents the player $i$ choosing how it interacts with the mediator and how it uses those interactions to inform its choices of action. Then, as before, an equilibrium is a tuple $(\pi, \phi_1, \dots, \phi_n)$ where $\phi \in \Delta(\mc M_1 \times \dots \mc M_n)$ is a distribution over mediators and no player $i$ can profit by switching to a different $\phi_i' \in \Phi_i$. The three notions above then differ in the choice of set $\Phi_i$:
\begin{itemize}
    \item Extensive-form correlated equilibria are equilibria where $\Phi_i$ is the set of interactions in which the player, upon reaching a decision point $j$, must send that decision point to the mediator.
    \item Private communication equilibria are equilibria  where $\Phi_i$ is the set of interactions in which the player, upon reaching a decision point $j$, must send a single message (which may or may not be the decision point $j$) to the mediator.
    \item Untimed private communication equilibria are equilibria  where $\Phi_i$ is the set of interactions in which the player, upon reaching a decision point $j$, may send any number of messages to the mediator.
\end{itemize}

The {\em direct mediator} $M_i^{\vec {\vec x_i}}$ for a pure strategy $\vec x_i \in \mc X_i$ is the mediator who acts by sending the recommendation $a$ at infoset $j$ if and only if the message history matches the $\Root_i \to j$ path, otherwise $\bot$. Formally, $M_i^{\vec x_i}(\vec s) = a(\vec x_i, j)$ if $\vec s = (j^{(1)}, a^{(1)}, j^{(2)}, \dots, j)$ is the path to $j$ in player $i$'s decision tree, and $\bot$ otherwise. We write $\mc M_i^* := \{ M_i^{\vec x_i} \mid  \vec x_i \in \mc X_i\}$  for the set of direct mediators on $\mc X_i$. Notice that, for direct $M_i$, the sets of interactions valid for each of the three equilibrium notions reduces to the sets of deviations defined in \Cref{sec:mediator}. Analogous to the NFCE case, an equilibrium $(\pi, \phi_1^*, \dots, \phi_n^*)$ (in any of the previous three notions) is called {\em direct} if $\pi$ is a distribution over direct mediators, and $\phi_i^*$ is the map $M_i^{\vec x_i} \mapsto \vec x_i$ (which is the analogy of the identity map). We are now ready to state the revelation principle for these notions.
\begin{theorem}[Revelation principle: for EFCE, proven by \citet{Stengel08:Extensive}; for communication equilibrium, proven by \citet{Myerson86:Multistage,Forges86:Approach} and refined by \citet{Zhang22:Polynomial}]
    For EFCE and (private) communication equilibrium, every equilibrium is outcome-equivalent to a direct equilibrium.
\end{theorem}
Our main result in this section is that the same holds for untimed private communication equilibrium:
\begin{restatable}[Revelation principle for untimed private communication equilibrium]{theorem}{thRP}\label{th:rp}
    Every untimed private communication equilibrium is outcome-equivalent to a direct untimed private communication equilibrium.
\end{restatable}
\begin{proof}
    Let $(\pi, \phi_1, \dots, \phi_n)$ be some (possibly indirect) equilibrium. Observe that we can view the mediator as holding a strategy $\vec y \in \mc Y_i$, where $\mc Y_i$ is the decision problem whose nodes correspond to sequences $\vec s \in \mc S_i^{\le H}$, \ie, to message histories. Notice that, by construction of the message set $\mc S_i$, $\mc Y$ contains a copy of each $\mc X_i$ within it, and that direct mediators $M_i^{\vec x_i}$ constrain themselves to states within this copy of $\mc X_i$ by terminating the interaction (sending $\bot$ forever) if the history of communication fails to match a state in player $i$'s decision problem. We will use this fact later.

    By \Cref{th:general pair}, each player's strategy set $\Phi_i$ is the set of linear maps $\mc Y_i \to \co \mc X_i$. Now, consider the direct profile $(\pi^*, \phi_i^*, \dots, \phi_n^*)$ where $\pi^* \in \Delta(\mc M_1^* \times \dots \times \mc M_n^*)$ is given by sampling $(M_1, \dots, M_n) \sim \pi$, sampling $\vec x_i \in \mc X_i$ from any distribution whose expectation is $\phi_i(M_i)$ for every player $i$, and finally outputting $(M_1^{\vec x_1}, \dots, M_n^{\vec x_n})$. Clearly, this profile is outcome-equivalent to the original profile, so it only remains to show that it is also an equilibrium. Consider any deviation $\phi'_i$ of player $i$ from the direct equilibrium.

    We proceed by contrapositive. Suppose that $(\pi^*, \phi_i^*, \dots, \phi_n^*)$ is not an equilibrium: player $i$ has profitable deviation $\phi'_i$. Since a direct mediator is constrained, as above, to act within player $i$'s decision problem, $\phi'_i$ can be expressed as a UTC deviation $\phi'_i : \mc X_i \to \co \mc X_i$. Since all UTC deviations are linear, $\phi'_i$ is itself linear, and can also be extended to a function $\phi'_i : \co \mc X_i \to \co \mc X_i$. Now let $\psi_i : \mc Y_i \to \co \mc X_i$ be given by $\psi_i = \phi'_i \circ \phi_i$, and observe that, since the composition of linear functions is linear, $\psi_i$ is a linear map, that is, $\psi_i \in \Phi_i$. Moreover, by construction, the profiles $(\pi, \psi_i, \phi_{-i}$) and $(\pi^*, \phi_i', \phi_{-i}^*)$ must induce the same outcome distributions---and therefore, $\psi_i$ is a profitable deviation against the original equilibrium $(\pi, \phi_1, \dots, \phi_n)$.
\end{proof}
This result justifies the definitions of equilibrium we have been using throughout the paper before reaching this point. We remark that, although the proof is usually not difficult, the revelation principle is not a given or automatic fact that can be assumed without proof: there are other settings where it fails, such as when the deviator's set of allowable messages depends on its true type in a nontrivial manner (\eg, \citealp{Forges05:Communication,Kephart21:Revelation}).

\section{Previous $\Phi$-Regret Algorithms}\label{sec:related}

\begin{table}[H]
    \renewcommand{\arraystretch}{1.3}
    \centering
    \scalebox{1}{
    \begin{tabular}{llccl}
    \toprule
       \bf Citation  & \bf Deviation set $(\Phi)$ & \bf Regret bound & \bf Complexity & \bf CFR? \\ 
    \midrule
        \citet{Zinkevich07:Regret} & External & $O(d \sqrt{T})$ & $O(d)$ & Yes \\
        \citet{Farina22:Kernelized} & External & $O(\sqrt{dT})$ & $O(d)$ & No \\
        \citet{Farina22:Simple} & Trigger & $O(d \sqrt{T})$ & FP & Yes\\
        \citet{Fujii23:Bayes}\,$^\ddagger$ & Communication & $\tilde O(\sqrt{dT})$ & FP & No \\
        \citet{Farina23:Polynomial} & Linear & $O(d^2 \sqrt{T})$ & QP & No \\
        {\bf This paper} & Linear & $O(d^2 \sqrt{T})$ & FP & Yes \\[1mm]
        \parbox{4cm}{\citet{Peng23:Fast}\\\citet{Dagan23:External}} & Swap & $T \cdot \tilde O\qty(\frac{\log d}{\log T})$ & $O(d \log T)$ & Yes\\
    \bottomrule
    \end{tabular}
    }
    \caption{Comparison of $\Phi$-regret minimizing algorithms for extensive-form games. {\bf Complexity} is per-iteration. ``QP'' and ``FP'' denote solving a quadratic program and a fixed-point problem, respectively.   {\bf CFR?} denotes whether the algorithm is based on the CFR framework, which is important in practice because, as stated in the body, CFR-based methods are the best practical regret minimizers. $^\ddagger$: The algorithm and analysis of \citet{Fujii23:Bayes} only applies to single-step Bayesian games, not general extensive-form games.}
    \label{tab:phi-regret-algs}
\end{table}

\changed{ \Cref{tab:phi-regret-algs} summarizes the various $\Phi$-regret algorithms known for extensive-form games. As suggested by the table, the main improvements of our algorithm compared to that of \citet{Farina23:Polynomial} in the same setting are (1) the faster per-iteration complexity (fixed-point computation versus quadratic program), and (2) the enabling of the use of the CFR framework, which leads to faster practical performance. The remarkable result of \citet{Peng23:Fast} and \citet{Dagan23:External} is very recent and in parallel with our work, and shows that there exists a PTAS for NFCE (with runtime roughly $d^{\tilde O(1/\varepsilon)}$). It remains an open problem  whether there exists a swap-regret minimizing algorithm whose regret  is $\poly(d) \cdot T^c$ for some $c < 1$, which would imply an FPTAS for NFCE in extensive-form games.}

\section{Another Example}\label{sec:more-examples}

\input{figures/example2}
In this section, we provide another example of untimed communication deviations, especially as they differ from (timed) communication deviations. Consider the game in \Cref{fig:example2}. Consider the correlated profile that mixes uniformly between the pure profiles (\action{1}{a}{1}, \action{1}{b}{1}, \action{2}{c}{1}) and (\action{1}{a}{1}, \action{1}{b}{2}, \action{2}{c}{2}). This is a communication equilibrium: \pone cannot profitably deviate, because its utility in the \util{1}{B} subgame is always 0, and if it chooses to disobey the recommendation \action{1}{a}{1} its expected utility will be also 0, because it cannot ask for another recommendation before choosing what action to play. However, \pone has the following profitable UTC deviation: ask for the recommendation at \util{1}{B} {\em before} deciding which action to play at \util{1}{A}. If the recommendation is \action{1}{b}{1}, play \action{1}{a}{2}; if the recommendation is \action{1}{b}{2}, play \action{1}{a}{3}.

Notice that, in this example, \pone's decision problem is essentially that of a normal-form game; therefore, its linear deviations coincide with its swap deviations. However, due to the timing restriction on communication deviations, the communication deviations are more restricted than the swap deviations.

This example also shows that the untimed private communication equilibria (see \Cref{sec:revelation}) are not outcome-equivalent to the timed private communication equilibria: in this game, every correlated profile is a distribution over terminal nodes (outcomes), so the fact that there exists a private communication equilibrium with a profitable UTC deviation is enough to disprove outcome equivalence.

\section{Proof of \Cref{th:main}}\label{app:proofs}
\thMain*
We start with a lemma.

\begin{lemma}\label{lem:representation}
    Let $f : \mc X \to \R_{\ge 0}$ be a linear map, where $\mc X$ is a sequence-form strategy space. Then there exists a unique vector $\vec c$ such that:
    \begin{enumerate}
        \item $f(\vec x) = \vec c^\top \vec x$ for all $\vec x \in \mc X$,
        \item $\vec c$ has all nonnegative entries, and
        \item for every decision point $I$, there is at least one action $a$ such that $\vec c(ja) = 0$.
    \end{enumerate}
\end{lemma}
\begin{proof}
    Let $f(\vec x) = \vec c^\top \vec x$, where $\vec c$ is currently arbitrary (\ie, it may not satisfy (2) and (3)). Then, for each decision point $j$ in bottom-up order, let $\vec c^*(j) := \min_a \vec c(ja)$. Subtract $\vec c^*(j)$ from $\vec c(ja)$ for every action $a$, and add $\vec c^*(j)$ to $\vec c(p_j)$. Since $\vec x$ satisfies the constraint $\vec x(p_j) = \sum_a x(ja)$, this does not change the validity of $\vec c$, and by the end of the algorithm, (2) and (3) will be satisfied except possibly that $\vec c(\Root) \ge 0$. To see that $\vec c(\Root) \ge 0$, let $\vec x$ be the pure strategy that plays the zeroing action $a$ specified by (3) at every decision point. Then, by construction, $\vec c^\top \vec x = \vec c(\Root) \ge 0$. To see that $\vec c$ is unique, note that there was no choice at any step in the above process: the transformation performed at each decision point is the only way to satisfy conditions (2) and (3) without changing the linear map.
\end{proof}
Now let $\mat A$ represent a linear map $\mc X \to \mc X$, where the rows of $\mat A$ are represented according to the above lemma. That is, $\mat A$ has all nonnegative entries, and moreover for any $\tilde\jmath \in \mc J$ and $\sigma \in \Sigma$, we have  $\mat A(\sigma, \tilde\jmath a) = 0$ for some action $a$. It remains only to show:
\begin{lemma}
    There exists a matrix $\mat B$ such that $(\mat A, \mat B)$ satisfies all constraints in the constraint system \eqref{eq:constraint system}.
\end{lemma}
\begin{proof}
    $\mat A(\Root, \Root) = 1$ and $\mat A(\Root, \tilde\sigma) = 0$ for $\tilde\sigma \ne \Root$ follow from the fact that $(\mat A \vec x)(\Root) = 1$ for all $\vec x$, that is $\mat A(\Root, \cdot) : \mc X \to [0, 1]$ is the identically-$1$ function, which by \Cref{lem:representation} has the above form. We are thus left with the main constraint,
    \begin{align}\label{cons:main}
        \mat A(p_j, \tilde{\sigma}) + \mat B(j, p_{\tilde \sigma}) & = \sum_{a \in A_j} \mat A(ja, \tilde\sigma) + \sum_{\tilde\jmath \in C_{\tilde\sigma}} \mat B(j, \tilde\jmath)
    \end{align}
    for every $(j, \tilde \sigma) \in \mc J \times \Sigma$.

    Define $\mat B$ and another matrix $\mat{\tilde B} \in \R^{\mc J\times \Sigma}$ as follows:
    \begin{align}
        \mat{\tilde B}(j, \tilde\sigma) & = \sum_{a \in A_j} \mat A(ja, \tilde\sigma) + \sum_{\tilde\jmath \in C_{\tilde\sigma}} \mat B(j, \tilde\jmath) - \mat A(p_j, \tilde\sigma) \\
        \mat B(j, \tilde\jmath)         & = \min_{a \in A_{\tilde\jmath}} \mat{\tilde B}(j, \tilde\jmath a)
    \end{align}
    To see that $\mat B$ satisfies all the constraints \eqref{cons:main}, let $\vec x$ be any fully-mixed strategy, and $I$ be any decision point. Then:
    \begin{align}
        0 & = \sum_{a \in A_j} (\mat A \vec x)(ja) - (\mat A \vec x)(p_j)
        \\&= \sum_{\tilde\sigma \in \Sigma} \vec x(\tilde\sigma) \qty(\sum_{a \in A_j} \mat A(ja, \tilde\sigma) - \mat A(p_j, \tilde\sigma))
        \\&= \sum_{\tilde\sigma \in \Sigma}\vec x(\tilde\sigma) \qty(\mat{\tilde B}(j, \tilde\sigma) - \sum_{\tilde\jmath \in C_{\tilde\sigma}} \mat B(j, \tilde\jmath))
        \\&= \sum_{\tilde\jmath}\qty(\sum_{a \in A_{\tilde\jmath}} \vec x(\tilde\jmath a) \mat{\tilde B}(j, \tilde\jmath a) - \vec x(p_{\tilde\jmath})\mat B(j, \tilde\jmath))
        \\&\ge \sum_{\tilde\jmath}\mat B(j, \tilde\jmath) \qty(\sum_{a \in A_{\tilde\jmath}} \vec x(\tilde\jmath a) - \vec x(p_{\tilde\jmath})) = 0
    \end{align}
    Thus, the inequality must in fact be an equality, and since all its terms are nonnegative, we thus have $\mat{\tilde B}(j, \tilde\jmath a) = \mat B(j, \tilde\jmath)$ for all $a \in A_{\tilde\jmath}$, so the constraints \eqref{cons:main} are satisfied by definition of $\mat{\tilde B}$.

    To see that $\mat B \ge 0$, suppose not. Let $(j, \tilde\jmath)$ be a last (\ie, farthest from the root, with respect to the ordering of the DAG) pair for which $\mat B(j, \tilde\jmath) < 0$. Then, for any action $\tilde a \in A_{\tilde\jmath}$, we have
    \begin{align}
        \mat A(p_j, \tilde\jmath \tilde a) + \mat B(j, \tilde\jmath) & = \sum_{a \in A_j} \mat A(ja, \tilde\jmath \tilde a) + \sum_{\tilde\jmath' \in C_{\tilde\jmath \tilde a}} \mat B(j, \tilde\jmath') \ge 0
    \end{align}
    where the inequality is because $(j, \tilde\jmath)$ is farthest from the root so all the terms on the right-hand side are nonnegative. Therefore, $\mat A(p_j, \tilde\jmath \tilde a) > 0$. But this should hold for every action $\tilde a$, contradicting the construction of $\mat A$, which includes the condition that there must exist a $\tilde a$ for which $\mat A(p_j, \tilde\jmath \tilde a) = 0$.
\end{proof}

\section{Additional Plots}\label{app:plots}

Below we present a variant of the plots in \Cref{fig:experiments}, in which \emph{time} (and not iterations) is reported on the x-axis.

\begin{figure}[H]
    \resizebox{.5\textwidth}{!}{\import{plots}{K45_time.pgf}}%
    \resizebox{.5\textwidth}{!}{\import{plots}{U222_time.pgf}}\\[-4mm]
    \resizebox{.5\textwidth}{!}{\import{plots}{L3132_time.pgf}}%
    \resizebox{.5\textwidth}{!}{\import{plots}{S2123_time.pgf}}\\[-1cm]
    \caption{Experimental comparison between our dynamics and those of \citet{Farina23:Polynomial} for approximating a linear correlated equilibrium in extensive-form games. Each algorithm was run for a maximum of $100,\!000$ iterations or 6 hours, whichever was hit first. Runs that were terminated due to the time limit are marked with a square {\small $\blacksquare$}. Compared to \Cref{fig:experiments}, the plots in this figure have \emph{time} on the x-axis.}
    \label{fig:experiments xtime}
\end{figure}

\end{document}

%% file: figures/example.tex
\begin{figure*}[p]
    \tikzset{
        every path/.style={-},
        every node/.style={draw},
        infoset1/.style={-, dotted, ultra thick, color=p1color},
        infoset2/.style={infoset1, color=p2color},
        terminal/.style={},
      ilabel/.style={fill=white, inner sep=0pt, draw=none},
    }
    \forestset{
            default preamble={for tree={
            parent anchor=south, child anchor=north, anchor=center, s sep=14pt, l=36pt
    }},
      p1/.style={
          regular polygon,
          regular polygon sides=3,
          inner sep=2pt, fill=p1color, draw=none},
      p2/.style={p1, shape border rotate=180, fill=p2color},
      parent/.style={no edge,tikz={\draw (#1.parent anchor) to (!.child anchor);}},
      parentd/.style n args={2}{no edge,tikz={\draw[ultra thick, p#2color] (#1.parent anchor) to (!.child anchor); }},
      nat/.style={},
      terminal/.style={draw=none, inner sep=2pt},
      el/.style n args={3}{edge label={node[midway, fill=white, inner sep=1pt, draw=none, font=\sf\footnotesize] {\action{#1}{#2}{#3}}}},
      d/.style={edge={ultra thick, draw={p#1color}}},
      comment/.style={no edge, draw=none, align=center, font=\tiny\sf},
  }
  \begin{center}
  \begin{forest}
    [,nat
      [,p1,label=right:{\util1{A}} [\util1{0},terminal,el={1}{a}{1}] [\util1{0},terminal,el={1}{a}{2}]]
      [,p1,label=right:{\util1{B}} [\util1{0},terminal,el={1}{b}{1}] [\util1{0},terminal,el={1}{b}{2}]]
      [,p1,label=right:{\util1{C}}
        [\util1{0},terminal,el={1}{c}{1}]
        [,nat,el={1}{c}{2}
            [,p1,label=right:{\util1{D}}
                [,p2,name=Fa,el={1}{d}{1} [\util1{+1},terminal, el={2}{f}{1}] [\util1{--10},terminal, el={2}{f}{2}]]
                [,p2,name=Fb,el={1}{d}{2} [\util1{--10},terminal, el={2}{f}{2}] [\util1{+1},terminal, el={2}{f}{2}]]
            ]
            [,p1,label=right:{\util1{E}}
                [,p2,name=Ga,el={1}{e}{1} [\util1{+1},terminal, el={2}{g}{1}] [\util1{--10},terminal, el={2}{g}{2}]]
                [,p2,name=Gb,el={1}{e}{2} [\util1{--10},terminal, el={2}{g}{2}] [\util1{+1},terminal, el={2}{g}{2}]]
            ]
        ]
      ]
    ]
    \draw[infoset2] (Fa) -- node[midway, ilabel]{\util2{F}} (Fb);
    \draw[infoset2] (Ga) -- node[midway, ilabel]{\util2{G}} (Gb);
  \end{forest}
  \end{center}
   \caption{An example extensive-form game in which communication deviations are a strict subset of UTC deviations. There are two players, P1 (\pone) and P2 (\ptwo). Infosets for both players are labeled with capital letters (\eg, \util{1}{A}) and joined by dotted lines. Actions are labeled with lowercase letters and subscripts (\eg, \action{1}{a}{1}). P1's utility is labeled on each terminal node. P2's utility is zero everywhere (not labeled). Boxes are chance nodes, at which chance plays uniformly at random.}\label{fig:example}
\end{figure*}
\begin{figure*}[p]
\tikzset{
        every path/.style={-},
        every node/.style={draw},
        infoset1/.style={-, dotted, ultra thick, color=p1color},
        infoset2/.style={infoset1, color=p2color},
        terminal/.style={},
      ilabel/.style={fill=white, inner sep=0pt, draw=none},
    }
    \forestset{
            default preamble={for tree={
            parent anchor=south, child anchor=north, anchor=center, s sep=14pt, l=36pt
    }},
      p1/.style={
          regular polygon,
          regular polygon sides=3,
          inner sep=2pt, fill=p1color, draw=none},
      p2/.style={p1, shape border rotate=180, fill=p2color},
      parent/.style={no edge,tikz={\draw (#1.parent anchor) to (!.child anchor);}},
      parentd/.style n args={2}{no edge,tikz={\draw[ultra thick, p#2color] (#1.parent anchor) to (!.child anchor); }},
      nat/.style={},
      terminal/.style={draw=none, inner sep=2pt},
      el/.style n args={3}{edge label={node[midway, fill=white, inner sep=1pt, draw=none, font=\sf\footnotesize] {\action{#1}{#2}{#3}}}},
      d/.style={edge={ultra thick, draw={p#1color}}},
      comment/.style={no edge, draw=none, align=center, font=\tiny\sf},
  }
  \begin{flushright}
  \begin{forest}
      [,nat
        [,p1,el={1}{A}{}
            [,terminal,el={1}{a}{1},d=0]
            [,terminal,el={1}{a}{2}]
            [...,terminal,el={3}{A}{}]
            [...,terminal,el={3}{B}{}]
            [...,terminal,el={3}{C}{}]
        ]
        [,p1,el={1}{B}{}
            [,terminal,el={1}{b}{1},d=0]
            [,terminal,el={1}{b}{2}]
            [...,terminal,el={3}{A}{}]
            [...,terminal,el={3}{B}{}]
            [...,terminal,el={3}{C}{}]
        ]
        [,p1,el={1}{C}{}
            [,terminal,el={1}{c}{1}]
            [,el={1}{c}{2},d=0
                [,p1,el={1}{D}{}
                    [,terminal,el={1}{d}{1}]
                    [,terminal,el={1}{d}{2}]
                    [,el={3}{A}{},d=0
                        [,p1,el={3}{a}{1}
                            [,terminal,el={1}{d}{1},d=0]
                            [,terminal,el={1}{d}{2}]
                        ]
                        [,p1,el={3}{a}{2}
                            [,terminal,el={1}{d}{1}]
                            [,terminal,el={1}{d}{2},d=0]
                        ]
                    ]
                    [...,terminal,el={3}{B}{}]
                    [...,terminal,el={3}{C}{}]
                ]
                [,p1,el={1}{E}{}
                    [,terminal,el={1}{e}{1}]
                    [,terminal,el={1}{e}{2}]
                    [...,terminal,el={3}{A}{}]
                    [,el={3}{B}{},d=0
                        [,p1,el={3}{a}{1}
                            [,terminal,el={1}{e}{1},d=0]
                            [,terminal,el={1}{e}{2}]
                        ]
                        [,p1,el={3}{a}{2}
                            [,terminal,el={1}{e}{1}]
                            [,terminal,el={1}{e}{2},d=0]
                        ]
                    ]
                    [...,terminal,el={3}{C}{}]
                ]
            ]
            [...,terminal,el={3}{A}{}]
            [...,terminal,el={3}{B}{}]
            [...,terminal,el={3}{C}{}]
        ]
      ]
  \end{forest}
  \end{flushright}
  \vspace{-5cm}
  {\footnotesize
  \setlength{\tabcolsep}{1pt}
  \setlength\extrarowheight{1pt}
  \adjustbox{valign=b}{\begin{tabular}{|cc|cc|cc|cc|cc|cc|}
  \hline
      & \action3{$\pmb\Root$}{} & \action3a1 & \action3a2 & \action3b1& \action3b2 & \action3c1 & \action3c2 & \action3d1 & \action3d2 & \action3e1 & \action3e2 \\
        \action1{$\pmb\Root$}{} &1&&&&&&&&&& \\\hline
        \action1a1 &1&&&&&&&&&& \\
        \action1a2 &&&&&&&&&&& \\\hline
        \action1b1 &1&&&&&&&&&& \\
        \action1b2 &&&&&&&&&&& \\\hline
        \action1c1 &&&&&&&&&&& \\
        \action1c2 &1&&&&&&&&&& \\\hline
        \action1d1 &&1&&&&&&&&& \\
        \action1d2 &&&1&&&&&&&& \\\hline
        \action1e1 &&&&1&&&&&&& \\
        \action1e2 &&&&&1&&&&&&\\
    \hline
  \end{tabular}}
  \hspace{-3pt}
  \adjustbox{valign=b}{\begin{tabular}{|c|ccc|cc|}
  \hline
      & \action3A{} & \action3B{} & \action3C{} & \action3D{} & \action3E{} \\\hline
      \action1A{} &&&&& \\
      \action1B{} &&&&& \\
      \action1C{} &&&&& \\\hline
      \action1D{} &1&&&& \\
      \action1E{} &&1&&& \\
    \hline
  \end{tabular}}}
 \caption{A part of the UTC decision problem for \pone corresponding to the same game. Nodes labeled \pone are decision points for \pone; boxes are observation points. ``...'' denotes that the part of the decision problem following that edge has been omitted. Terminal nodes are unmarked. Red edge labels indicate interactions with the mediator; blue edge labels indicate interactions with the game. The profitable untimed deviation discussed in \Cref{sec:example} is indicated by the thick lines. The first action taken in that profiable deviation, \action{1}{c}{2}, is not legal for a timed deviator, because a timed deviator must query the mediator once before taking its first action. The matrices (lower-left corner) are the pair of matrices $(\mat A, \mat B)$  corresponding to that same deviation. All blank entries are 0.}\label{fig:example-dev}
\end{figure*}

%% file: figures/trapezoid.tex
\begin{figure}[H]
    \centering
    \begin{tikzpicture}
        \draw (0,0) node[circle, fill=black, inner sep=1pt, label=below:{$A=\phi(D)$}, name=A] {};
        \draw (8,0) node[circle, fill=black, inner sep=1pt, label=below:{$B$}, name=B] {};
        \draw (4,2) node[circle, fill=black, inner sep=1pt, label={$C=\phi(B)$}, name=C] {};
        \draw (0,2) node[circle, fill=black, inner sep=1pt, label={$D=\phi(A)$}, name=D] {};
        \draw (2,0) node[circle, fill=black, inner sep=1pt, label=below:{$\phi(C)$}, name=fC] {};
        \draw [p1color,ultra thick] (A.center) -- (B.center) -- (C.center) -- (D.center) -- cycle; 
        \draw [p2color, thick] (D.center) -- (C.center) -- (fC.center) -- (A.center) -- cycle;
    \end{tikzpicture}
    \caption{A visual depiction of the argument that \Cref{cor:convhull} cannot generalize to all polytopes. The affine map $\phi$ maps the large blue polygon onto the small orange polygon, and $\phi$ is a vertex of the set of linear maps from polygon $ABCD$ to itself, yet $\phi(C)$ is not a vertex of $ABCD$.}\label{fig:trapezoid}
\end{figure}

%% file: figures/example2.tex
\begin{figure*}
    \tikzset{
        every path/.style={-},
        every node/.style={draw},
        infoset1/.style={-, dotted, ultra thick, color=p1color},
        infoset2/.style={infoset1, color=p2color},
        terminal/.style={},
      ilabel/.style={fill=white, inner sep=0pt, draw=none},
    }
    \forestset{
            default preamble={for tree={
            parent anchor=south, child anchor=north, anchor=center, s sep=14pt, l=36pt
    }},
      p1/.style={
          regular polygon,
          regular polygon sides=3,
          inner sep=2pt, fill=p1color, draw=none},
      p2/.style={p1, shape border rotate=180, fill=p2color},
      parent/.style={no edge,tikz={\draw (#1.parent anchor) to (!.child anchor);}},
      parentd/.style n args={2}{no edge,tikz={\draw[ultra thick, p#2color] (#1.parent anchor) to (!.child anchor); }},
      nat/.style={},
      terminal/.style={draw=none, inner sep=2pt},
      el/.style n args={3}{edge label={node[midway, fill=white, inner sep=1pt, draw=none, font=\sf\footnotesize] {\action{#1}{#2}{#3}}}},
      d/.style={edge={ultra thick, draw={p#1color}}},
      comment/.style={no edge, draw=none, align=center, font=\tiny\sf},
  }
  \begin{center}
  \begin{forest}
      [,p1,label=right:{\util1{A}}
        [,p1,label=right:{\util1{B}},el={1}{a}{1}
            [,p2,name=D1,tier=1,el={1}{b}{1}
                [\util1{0},terminal,el={2}{c}{1}]
                [\util1{0},terminal,el={2}{c}{2}]
            ]
            [,p2,name=D2,tier=1,el={1}{b}{2}
                [\util1{0},terminal,el={2}{c}{1}]
                [\util1{0},terminal,el={2}{c}{2}]
            ]
        ]
        [,p2,name=D3,tier=1,el={1}{a}{2}
            [\util1{1},terminal,el={2}{c}{1}]
            [\util1{--1},terminal,el={2}{c}{2}]
        ]
        [,p2,name=D4,tier=1,el={1}{a}{3}
            [\util1{--1},terminal,el={2}{c}{1}]
            [\util1{1},terminal,el={2}{c}{2}]
        ]
      ]
      \draw[infoset2] (D1) -- node[midway, ilabel]{\util2{C}}(D4);
  \end{forest}
  \end{center}
  \caption{Another example. The notation is shared with \Cref{fig:example}. In this example, \pone's strategy set is equivalent to a simplex, so the linear deviations coincide with its swap deviations. As such, we will not bother to depict the UTC decision problem or matrices.}\label{fig:example2}
\end{figure*}